\documentclass[11pt, reqno]{amsart}

\usepackage[foot]{amsaddr}
\usepackage[utf8]{inputenc}   
\usepackage[english]{babel}   
\usepackage{enumerate}   
\usepackage{verbatim}   
\usepackage[margin=2.5cm]{geometry}   
\usepackage{amssymb}   
\usepackage{amsthm}   
\usepackage{amsmath}   
\usepackage{babel}   
\usepackage{csquotes}   
\usepackage{graphicx}   
\usepackage{caption}   
\usepackage{listings}   
\usepackage{bbm}   
\usepackage{todonotes}   
\usepackage{color,soul}   
\usepackage{tabularx,colortbl}   
\usepackage{tcolorbox}   
\usepackage{mathrsfs}
\usepackage{mathtools}
\usepackage[makeroom]{cancel}
\usepackage{tikz}
\usepackage{subcaption}
\usepackage[labelformat=parens,labelsep=quad,skip=3pt]{caption}
\usepackage{graphicx}
\usepackage{dashrule}

\usepackage{multirow}

\theoremstyle{plain}

\input amssym.def   
\input amssym.tex   
 
\def\beq{\begin{equation}}   
\def\eeq{\end{equation}}   
\def\beqn{\begin{eqnarray}}

\def\erf{\text{erf}} 

\def\eeqn{\end{eqnarray}}

\def\utilde#1{\mathord{\vtop{\ialign{##\crcr   
$\hfil\displaystyle{#1}\hfil$\crcr\noalign{\kern1.5pt\nointerlineskip}   
$\hfil {}\hfil$\crcr\noalign{\kern1.5pt}}}}}   
 
\theoremstyle{definition}   
\newtheorem{theorem}{Theorem}[section]

\newtheorem*{proposition*}{Proposition}   
   
\theoremstyle{remark}

\newtheorem*{example*}{\bf Example}   
\newtheorem{example}{\bf Example}

\theoremstyle{definition}

\def\R{\mathbb R}

\def\to{\rightarrow}   
\def\dlim[#1][#2]{\lim_{#1 \to #2, #1 \neq #2}}

\title{Comparison Between Bayesian and Frequentist Tail Probability Estimates}

\author{Nan Shen} 
\address{  Department of Mathematics, University of Utah, 155 South 1400 East, JWB 233
Salt Lake City, UT 84112.  \textup{\tt nan.shen@utah.edu}}

\author{Bárbara González-Ar\'evalo}
\address{College of Liberal Arts and Sciences, Northern Illinois University, 1425 Lincoln Hwy, DeKalb, IL 60115. \textup{\tt bgonzalez4@niu.edu}}

\author{Luis Ra\'ul Pericchi}
\address{ Department of Mathematics, University of Puerto Rico, Rio Piedras Campus, Box 70377  San Juan, PR 00936-8377. \textup{\tt luis.pericchi@upr.edu}}

\begin{document}

\maketitle

\begin{abstract}
    Tail probability plays an important part in the extreme value theory. Sometimes the conclusions from two approaches for estimating the tail probability of extreme events, the Bayesian and the frequentist methods, can differ a lot. In 1999, a rainfall that caused more than 30,000 deaths in Venezuela was not captured by the simple frequentist extreme value techniques. However, this catastrophic rainfall was not surprising if the Bayesian inference was used to allow for parameter uncertainty and the full available data was exploited \cite{cp}.
    
    In this paper, we investigate the reasons that the Bayesian estimator of the tail probability is always higher than the frequentist estimator. Sufficient conditions for this phenomenon are established both by using Jensen's Inequality and by looking at Taylor series approximations, both of which point to the convexity of the distribution function.
    
    \vspace{5mm}
    
\noindent \textbf{Keywords.} Tail Probability, Taylor Series, Jensen's Inequality, Convexity

\end{abstract}

\section{Introduction}
Tragedies like the 9/11 attacks, earthquakes or volcanic eruptions are rare events, but they are always followed by catastrophic consequences. Estimating the probabilities of extreme events has become more important and urgent in recent decades \cite{terr}. Both large deviations theory \cite{timo} and extreme value theory, which is widely used in disciplines like actuarial science, environmental sciences and physics\cite{book3}, investigate both the theoretical and practical problems arising from rare events \cite{book1} \cite{book2}. 

With the popularization of Bayesian statistics, we now have two approaches for evaluating the probability of the tail: Bayesian and frequentist \cite{smith} \cite{bandf} \cite{bayes}. However, these two methods sometimes give different conclusions. As the case shows in \cite{cp}, before 1999 simple frequentist extreme value techniques were used to predict future levels of extreme rainfall in Venezuela. In December 1999, daily precipitation of more than 410 mm, almost three times the daily rainfall measurements of the previously recorded maximum, was not captured which caused an estimated 30,000 deaths.   Figure \ref{cp1} in \cite{cp} shows that the precipitation of 1999 is exceptional even relative to the better fitting model under the frequentist MLE method.  However, Figure \ref{cp2} gives that the 1999 event can be anticipated if we use Bayesian inferences to fully account for the uncertainties due to parameter estimation and exploit all available data. Table \ref{table:1} which is taken from \cite{cp2} gives the return level estimates of 410.4 mm, the 1999 annual maximum, using different models. From the table we can also see that the Bayesian method gives a way much smaller return period that the frequentist Maximum Likelihood Estimation(MLE) under different models.

\begin{table}[h!]
\centering
\begin{tabular}{ |c|c|c|c| } 
\hline
\multicolumn{4}{|c|}{Return period of 410.4 mm} \\
\hline
Mode of Inference & Model & 1999 data excluded  & 1999 datum included\\
\hline
\multirow{2}{4em}{MLE} & Gumbel & 17,600,000 & 737,000\\ 
& GEV & 4280 & 302 \\ 
\hline
\multirow{2}{4em}{Bayes} & Gumbel & 2,170,000 & 233,000\\ 
& GEV & 660 & 177 \\ 
\hline
\end{tabular}
\caption{Return level estimates of 410.4 mm, the 1999 annual maximum, using different models
and modes of inference. For the MLE analysis the values correspond to the maximum likelihood
estimates of the return period. For the Bayesian analysis the values are the predictive return periods. From \cite{cp2} (with permission)}
\label{table:1}
\end{table}

\begin{figure}
\centering\includegraphics[width=10cm]{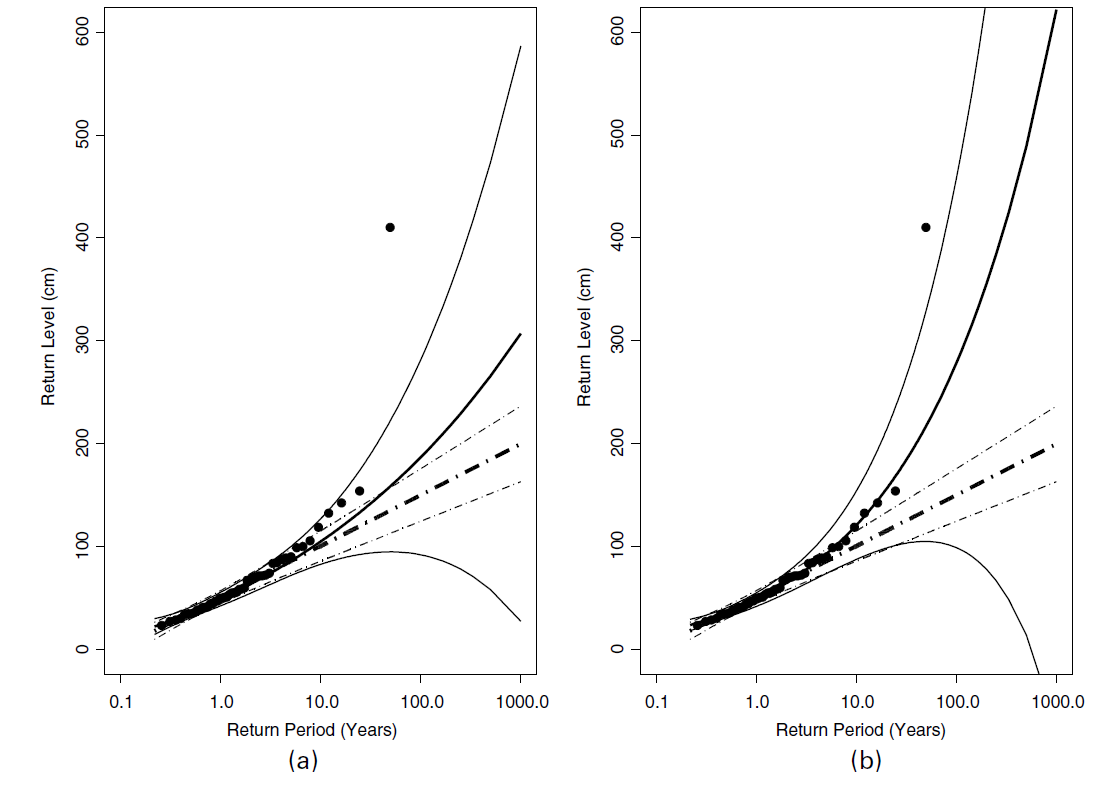}
\caption{Return level plots for models fitted to the annual maximum Venezuelan rainfall data( $\rule[0.5mm]{10mm}{0.3mm}$, GEV model, maximum likelihood estimates; $\rule[0.5mm]{10mm}{0.1mm}$, GEV model, limits of the 95$\%$ confidence intervals; $\hdashrule[0.5ex]{2cm}{0.4mm}{2.5mm 3pt 0.3mm 2pt}$, Gumbel model, maximum likelihood estimates; $\hdashrule[0.5ex]{1.5cm}{0.1mm}{1.5mm 3pt 0.3mm 2pt}$, Gumbel model, limits of the 95$\%$ confidence intervals; $\bullet$, empirical estimates based on the complete 49-year data set): (a) excluding the 1999 data; (b) including the 1999 data. }
\label{cp1}
\end{figure}

\begin{figure}
\centering\includegraphics[width=10cm]{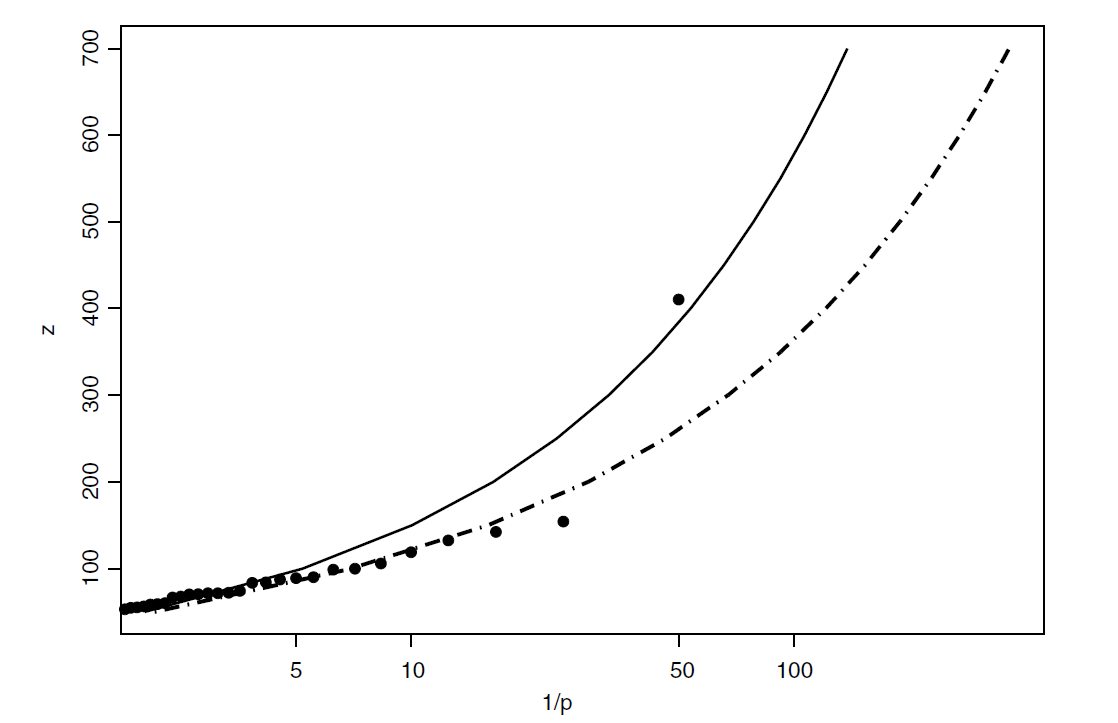}
\caption{Predictive distributions for seasonal models fitted to the Venezuelan rainfall data: $\rule[0.5mm]{7mm}{0.1mm}$, including the 1999 data; $\hdashrule[0.5ex]{2cm}{0.4mm}{2.5mm 3pt 0.3mm 2pt}$, excluding the 1999 data; $\bullet$, empirical estimates based on the 49 available annual maxima.}
\label{cp2}
\end{figure}

The lesson we learned from \cite{cp2} and \cite{cp} is the motivation for our research. Reasons why these two methods give huge different results, especially why the Bayesian model usually gives a larger probability of the tail than the classic frequentist method need to be investigated here. The Bayesian estimation of probability of tails is well founded on Probability Theory: It is a marginal computation that integrates out the parameters of the tail. On the other hand, the ``plug-in'' insertion a point estimate of the parameters is obviously an ``ad-hoc'' procedure not based on probability calculus but on an approximation, that may be close to the correct calculation at the center of the range but that deteriorates spectacularly as we move away to the interesting tails, where extreme values occur.

\subsection{Definitions and Goal}
Let $X$ be the random variable. Suppose $X$ indicate the magnitude and intensity of an earthquake, then the tail probability $P(X > a)$ identifies the probability of an earthquake occurrence when $a$ is some value much greater than the mean.   The Bayesian estimator of this probability is defined as
$$P_B(X>a)=\int_{\boldsymbol{\Theta}}[1-F(a|\boldsymbol{\theta})]\pi(\boldsymbol{\theta}|a)d\boldsymbol{\theta},$$
which is the expectation of the tail probability $1-F(a|\boldsymbol{\theta})$ under the posterior distribution $\pi(\boldsymbol{\theta}|x)$ given $x=a$, where $\boldsymbol{\theta}$ denotes the parameters in the distribution function and could be one dimension or generalized to a high dimensional vector. The frequentist estimator is also called the plug-in estimator which is defined as
$$P_F(X>a)=1-F(a|\hat{\boldsymbol{\theta}}),$$ 
where $\hat{\boldsymbol{\theta}}$ is the Maximum Likelihood Estimator(MLE) of $\boldsymbol{\theta}$.

Numerical experimental results point to the fact that the asymptotic behavior of the Bayesian estimator of the tail probability is usually higher than the frequentist estimator. We will use a very simple hierarchical model as an example to illustrate this phenomenon. Suppose we have a sequence of random variable $\boldsymbol{x}=x_1, \cdots, x_n$ that follows the exponential distribution. The density of the exponential distribution is given by
$$f(x|\lambda) = \frac{1}{\lambda}e^{-x/\lambda}$$
where $x \ge 0$ and $\lambda > 0$. So the tail probability is
$$\varphi(\lambda)=1-F(a|\lambda)=\int^{\infty}_a f(x|\lambda)dx=e^{-\frac{a}{\lambda}}$$
The marginal distribution can be calculated as 
\begin{equation*}
m({\bf x})=\int^{\infty}_0f({\bf x}|\lambda)\pi_J(\lambda)d\lambda = \Gamma(n)\Big(\sum_{i=1}^{n} x_i \Big) ^{-n},
\end{equation*}
where we use Jeffereys prior as $\pi_J(\lambda) \propto 1/ \lambda$.
By which the posterior distribution is obtained as
$$\pi(\lambda|{\bf x})=\frac{1}{\lambda^{(n+1)}\Gamma(n)}\exp\left(-\frac{\sum^n_{i=1}x_i}{\lambda}\right)\left(\sum^n_{i=1}x_i\right)^n$$
Then the Bayesian estimator of the tail probability is 
$$P_B(X>a)=E[\varphi(\lambda)|\boldsymbol{x}]=\int^{\infty}_{0}\varphi(\lambda) \pi(\lambda|\boldsymbol{x})d\lambda=\left(1+\frac{a}{n\bar{x}}\right)^{-n}$$
And the frequentist estimator of the tail probability is
$$P_F(X>a)=\varphi(\hat{\lambda}=\bar{x})=e^{-\frac{a}{\bar{x}}}$$
Note that $P_B$ goes to zero at a polynomial rate, whereas $P_F$ goes to zero at an exponential rate. So, they are not even asymptotically equivalent, and thus it should be the case that:
$$\frac{P_B(X>a)}{P_F(X>a)} \to \infty\  \ \text{ as } a \to 
\infty$$
We conduct the numerical experiments and plot the ratio of $P_B$ and $P_F$. The results below show the ratio for different range of $a$. We can see that when $a$ is some moderate number between 50 and 100. The ratio is greater than 1 and increases slowly. However, when $a$ varies from 500 to 1000, the ratio goes to extremely large number very quick. In other words, $\frac{P_B(X>a)}{P_F(X>a)}\to \infty \text{ as } a \to \infty$ which implies that the Bayesian estimate is higher than the frequentist estimate asymptotically.

\begin{figure}
  \begin{subfigure}{8cm}
    \centering\includegraphics[width=8cm]{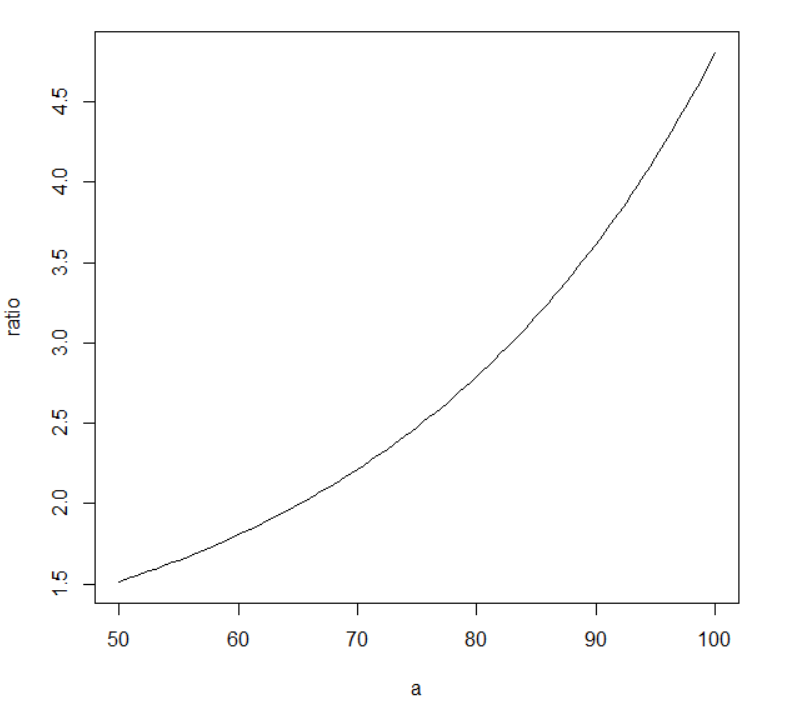}
    \caption{The ratio of $P_B$ and $P_F$ when $a$ is moderate.}
  \end{subfigure}
  \begin{subfigure}{8cm}
    \centering\includegraphics[width=8cm]{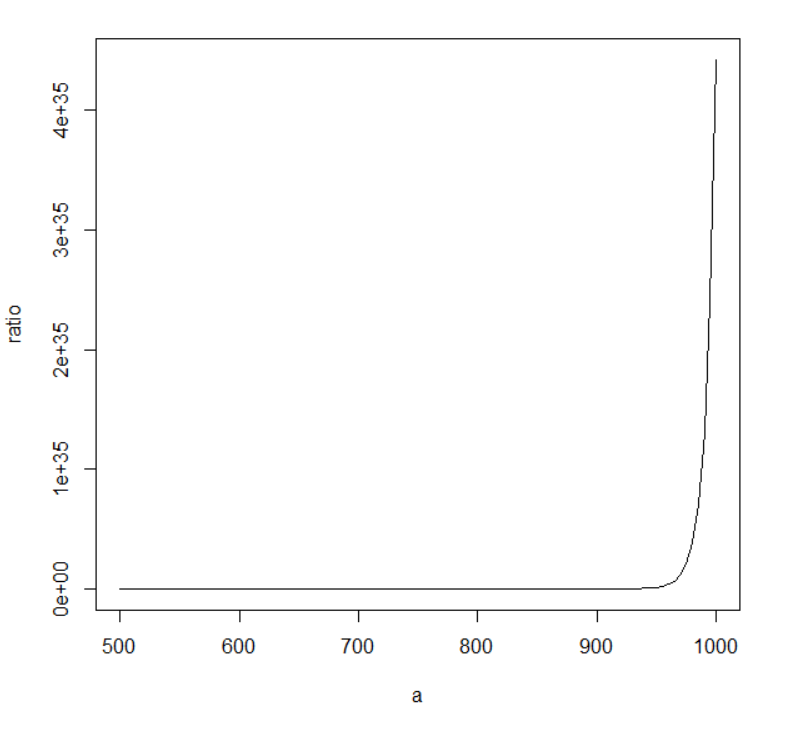}
    \caption{The ratio of $P_B$ and $P_F$ when $a$ is large.}
  \end{subfigure}
\caption{Plots of $\frac{P_B(X>a)}{P_F(X>a)}$ for different $a$. }
\end{figure}

We investigate the reasons for this behavior and the conditions under which it happens. We first use Jensen's inequality  \cite{convex} \cite{statstext}, which states that if $\varphi(\boldsymbol{\theta})=1-F(a|\boldsymbol{\theta})$ is a convex function of the random variable $\boldsymbol{\theta}$, then
$$\varphi \left(\operatorname {E} [\boldsymbol{\theta}]\right)\leq \operatorname {E} \left[\varphi (\boldsymbol{\theta})\right].$$

We then verify that the convexity conditions are met for certain distributions that are widely used in extreme value analysis using a Taylor series approximation of the tail probability $1-F(a|\boldsymbol{\theta})$ around the MLE of $\boldsymbol{\theta}$, and plug it into the difference between the Bayesian and the frequentist estimators which is defined as
$$D(a)=P_B(X>a)-P_F(X>a).$$

In conclusion, we can show that the convexity of $\varphi(\boldsymbol{\theta})=1-F(a|\boldsymbol{\theta})$ is a sufficient condition for $P_B > P_F$.  We also verify this convexity condition for specific distributions which are widely used in extreme value theory.

\section{Methodology}
We will investigate why the Bayesian estimator of the tail probability is usually asymptotically higher than the frequentist one in this section. The method is to prove that if $\varphi (\boldsymbol{\theta})=1-F(a|\boldsymbol{\theta})$ is convex then we can apply Jensen's inequality directly. Then we use the Taylor expansion of the tail probability $1-F(a|\boldsymbol{\theta})$ to verify the results we get which leads to the same conditions on our distribution function $F(a|\boldsymbol{\theta})$.

\subsection{Convexity Investigation Using Jensen's Inequality}
For tail probability estimations, Bayesian method gives $P_B(X>a)=E^{\pi(\boldsymbol{\theta}|a)}[1-F(a|\boldsymbol{\theta})]$, while frequentist method using $P_F(X>a)=1-F(a|\hat{\boldsymbol{\theta}})$. To investigate the relation between $P_B$ and $P_F$, Jensen's inequality tells something similar and I will state formally here as:
\begin{theorem}
Let $(\Omega ,A,\mu )$ be a probability space, such that $\mu (\Omega )=1$. If $g: \Omega \to \R^d$ is a real-valued function that is $\mu$-integrable, and if $\varphi$  is a convex function on the real line, then:
$$ \varphi \left(\int _{\Omega }g\,d\mu \right)\leq \int _{\Omega }\varphi \circ g\,d\mu .$$
\end{theorem}
Note here the measurable function $g$ is our parameter $\boldsymbol{\theta}$. So Jensen's inequality gives $ \varphi \left(\operatorname {E} [\boldsymbol{\theta}]\right)\leq \operatorname {E} \left[\varphi (\boldsymbol{\theta})\right]$ \cite{statstext}. The inequality we want to prove, however, is that $\varphi[\hat{\boldsymbol{\theta}}] \leq E[\varphi(\boldsymbol{\theta})]$. The following theorem and proof shows that as $a \to \infty$ we have $\varphi[\hat{\boldsymbol{\theta}}]$ and $\varphi[E(\boldsymbol{\theta})]$ are quite close to each other, which implies that $\varphi[\hat{\boldsymbol{\theta}}] \leq E[\varphi(\boldsymbol{\theta})]$.

\begin{theorem}
Let $X$ be a continuous random variable supported on semi-infinite intervals, usually $[c,\infty)$ for some $c$, or supported on the whole real line,  with $F(a|\boldsymbol{\theta})$ be the cumulative distribution function (CDF) of $X$ where $a$ is some extremely large number on the support, and $\varphi (\boldsymbol{\theta})=1-F(a|\boldsymbol{\theta})$ is a convex function. Suppose $\hat{\boldsymbol{\theta}}$ is the maximum likelihood estimation of the parameter $\boldsymbol{\theta}$, then $$\varphi[\hat{\boldsymbol{\theta}}] \leq  E[\varphi(\boldsymbol{\theta})]$$
\end{theorem}

\begin{proof}
$$\begin{aligned}
\left|\varphi[E(\boldsymbol{\theta})] - \varphi[\hat{\boldsymbol{\theta}}] \right| & = \left|\left(1-F[a|E(\boldsymbol{\theta})]\right)-\left(1-F[a|\hat{\boldsymbol{\theta}}]\right)\right|\\
& = \left|\left(1-\int^a_{-\infty} f(x|E(\boldsymbol{\theta}))dx \right)- \left(1-\int^a_{-\infty}  f(x|\hat{\boldsymbol{\theta}})dx \right)\right|\\
& =\left|\int^{\infty}_a  f(x|\hat{\boldsymbol{\theta}})dx - \int^{\infty}_a f(x|E(\boldsymbol{\theta}))dx \right|\\
& \leq \int^{\infty}_a \left|f(x|\hat{\boldsymbol{\theta}}) - f(x|E(\boldsymbol{\theta}))\right|dx
\end{aligned}$$
Let's define
$$g(x)=f(x|\hat{\boldsymbol{\theta}}) - f(x|E(\boldsymbol{\theta}))$$
Then we can see that
$$\begin{aligned}
\int^{\infty}_{-\infty}|g(x)| dx 
& \leq \int^{\infty}_{-\infty}\left|f(x|\hat{\boldsymbol{\theta}})\right| dx + \int^{\infty}_{-\infty}\left|f(x|E(\boldsymbol{\theta}))\right|dx \\
& =\int^{\infty}_{-\infty}f(x|\hat{\boldsymbol{\theta}}) dx + \int^{\infty}_{-\infty}f(x|E(\boldsymbol{\theta}))dx \\
& = 1+1 =2 < \infty.
\end{aligned}$$
Which means that $g(x)$ is a integrable function. Thus implies $\lim_{a \to \infty} \int^{\infty}_a |g(x)|dx =0 $, i.e
$$\lim_{a \to \infty} \int^{\infty}_a \left|f(x|\hat{\boldsymbol{\theta}}) - f(x|E(\boldsymbol{\theta}))\right|dx =0 $$
Thus, for $\forall \epsilon >0, \exists a$ such that
$$\int^{\infty}_a \left|f(x|\hat{\boldsymbol{\theta}}) - f(x|E(\boldsymbol{\theta}))\right|dx < \epsilon .$$
Which implies $\exists a$ such that $$|\varphi[\hat{\boldsymbol{\theta}}] - \varphi[E(\boldsymbol{\theta})]| < \epsilon .$$
Thus $- \epsilon < \varphi[\hat{\boldsymbol{\theta}}] -\varphi[E(\boldsymbol{\theta})]  < \epsilon $ or we could write
$$\varphi[E(\boldsymbol{\theta})] - \epsilon <  \varphi[\hat{\boldsymbol{\theta}}] < \varphi[E(\boldsymbol{\theta})] + \epsilon .$$
Equality in Jensen's inequality holds only if our function $\varphi$ is essentially constant,
and suppose our function $\varphi(\boldsymbol{\theta})$ is strictly convex, which is true for most of the cases that we encounter, then we know our Jensen's inequality is strict also, i.e. $\varphi[E(\boldsymbol{\theta})] < E[\varphi(\boldsymbol{\theta})]$. Which implies $\exists \epsilon >0$ such that $$E[\varphi(\boldsymbol{\theta})] \geq \varphi[E(\boldsymbol{\theta})] +\epsilon .$$ 
Hence for this $\epsilon$, as $a \to \infty$ we have
$$E[\varphi(\boldsymbol{\theta})] \geq \varphi[E(\boldsymbol{\theta})] +\epsilon > \varphi[\hat{\boldsymbol{\theta}}].$$
\end{proof}

\subsection{Taylor Expansion Examination}
In this section, we will use Taylor series for the tail probability to check the results we got in the previous section. Let $$D(a)=P_B(X>a)-P_F(X>a)=\int_{\boldsymbol{\Theta}}[1-F(a|\boldsymbol{\theta})]\pi(\boldsymbol{\theta}|a)d\boldsymbol{\theta}-\left(1-F(a|\hat{\boldsymbol{\theta}})\right)$$
which is the difference of the tail probabilities between the Bayesian and the frequentist estimators. The Taylor series of $1-F(a|\boldsymbol{\theta})$ at the MLE of $\boldsymbol{\theta}$, $\hat{\boldsymbol{\theta}}$, is given as
\begin{eqnarray}
\small
\varphi(\boldsymbol{\theta}) & = & 1-F(a|\boldsymbol{\theta})\\
& = & 1-F(a|\hat{\boldsymbol{\theta}})-\left. \nabla_{\boldsymbol{\theta}} F(a|\boldsymbol{\theta}) \right|_{\boldsymbol{\theta} =\hat{\boldsymbol{\theta}}}(\boldsymbol{\theta}-\hat{\boldsymbol{\theta}})-\frac{1}{2}\left. H_{\boldsymbol{\theta}}\left( F(a|\boldsymbol{\theta})\right) \right|_{\boldsymbol{\theta}=\hat{\boldsymbol{\theta}}}(\boldsymbol{\theta}-\hat{\boldsymbol{\theta}})^2-R(\boldsymbol{\theta})\\
R(\boldsymbol{\theta})&=&\frac{1}{6}\left.D^3_{\boldsymbol{\theta}}\left( F(a|\boldsymbol{\theta})\right) \right|_{\boldsymbol{\theta}={\boldsymbol{\theta}_L}}(\boldsymbol{\theta}-\hat{\boldsymbol{\theta}})^3 \text{ where }\boldsymbol{\theta}_L \text{ is between }\boldsymbol{\theta} \text{ and }\hat{\boldsymbol{\theta}}.
\end{eqnarray}
where $\nabla_{\boldsymbol{\theta}}F(a|\boldsymbol{\theta})$ is the gradient of $F(a|\boldsymbol{\theta})$ such that
$$\left(\nabla_{\boldsymbol{\theta}}F(a|\boldsymbol{\theta})\right)_i=\frac{\partial}{\partial \theta_i}F(a|\boldsymbol{\theta})$$
and $H_{\boldsymbol{\theta}}$ is the Hessian matrix of dimension $|\boldsymbol{\theta}|\times |\boldsymbol{\theta}|$ such that $$H_{ml}=\frac{\partial^2}{\partial \theta_m \partial \theta_l}F(a|\boldsymbol{\theta}) $$
And $D^3_{\boldsymbol{\theta}}\left( F(a|\boldsymbol{\theta})\right)$ is the third partial derivative of $F(a|\boldsymbol{\theta})$ w.r.t. $\boldsymbol{\theta}$ in a similar manner. Then $D(a)$ could be rewritten as
$$\small \begin{aligned}
D(a) 
& =\int_{\boldsymbol{\Theta}}[1-F(a|\boldsymbol{\theta})]\pi(\boldsymbol{\theta}|a)d\boldsymbol{\theta}-\left(1-F(a|\hat{\boldsymbol{\theta}})\right)\\
& = -\left.\nabla_{\boldsymbol{\theta}} F(a|\boldsymbol{\theta}) \right|_{\hat{\boldsymbol{\theta}}} \int_{\Theta}\pi(\boldsymbol{\theta}|a)(\boldsymbol{\theta}-\hat{\boldsymbol{\theta}})d\boldsymbol{\theta} - \frac{1}{2}\left. H_{\boldsymbol{\theta}}\left( F(a|\boldsymbol{\theta})\right) \right|_{\hat{\boldsymbol{\theta}}} \int_{-\infty}^\infty\pi(\boldsymbol{\theta}|a)(\boldsymbol{\theta}-\hat{\boldsymbol{\theta}})^2d\boldsymbol{\theta} -R^*(\boldsymbol{\theta})\\
& = -\left.\nabla_{\boldsymbol{\theta}} F(a|\boldsymbol{\theta}) \right|_{\hat{\boldsymbol{\theta}}} E^{\pi(\boldsymbol{\theta}|a)}(\boldsymbol{\theta}-\hat{\boldsymbol{\theta}}) - \frac{1}{2}\left. H_{\boldsymbol{\theta}}\left( F(a|\boldsymbol{\theta})\right) \right|_{\hat{\boldsymbol{\theta}}}E^{\pi(\boldsymbol{\theta}|a)}(\boldsymbol{\theta}-\hat{\boldsymbol{\theta}})^2 -R^*(\boldsymbol{\theta})
\end{aligned}$$
Here we simplify the notation by writing $d F(a|\boldsymbol{\theta})/ d \boldsymbol{\theta}\left.\right|_{\boldsymbol{\theta}=\hat{\boldsymbol{\theta}}}=F'(a|\hat{\boldsymbol{\theta}})$, $d^2 F(a|\boldsymbol{\theta})/ d \boldsymbol{\theta}^2 \left.\right|_{\boldsymbol{\theta}=\hat{\boldsymbol{\theta}}}=F''(a|\hat{\boldsymbol{\theta}})$, and
$$R^*(\boldsymbol{\theta})= \frac{1}{6}\left.D^3_{\boldsymbol{\theta}}\left( F(a|\boldsymbol{\theta})\right) \right|_{{\boldsymbol{\theta}_L}}\int_{-\infty}^\infty\pi(\boldsymbol{\theta}|x)(\boldsymbol{\theta}-\hat{\boldsymbol{\theta}})^3d\boldsymbol{\theta}=\frac{1}{6}\left.D^3_{\boldsymbol{\theta}}\left( F(a|\boldsymbol{\theta})\right) \right|_{{\boldsymbol{\theta}_L}} E^{\pi(\boldsymbol{\theta}|x)}(\boldsymbol{\theta}-\hat{\boldsymbol{\theta}})^3$$
In order for $\varphi(\boldsymbol{\theta})$ to be convex, and $D(a)$ to be negative we would need the first term\\ $\left.\nabla_{\boldsymbol{\theta}} F(a|\boldsymbol{\theta}) \right|_{\hat{\boldsymbol{\theta}}} E^{\pi(\boldsymbol{\theta}|a)}(\boldsymbol{\theta}-\hat{\boldsymbol{\theta}})$ and the third term $R^*(\boldsymbol{\theta})$ to go to zero asymptotically, and the second term $\left. H_{\boldsymbol{\theta}}\left( F(a|\boldsymbol{\theta})\right) \right|_{\hat{\boldsymbol{\theta}}}E^{\pi(\boldsymbol{\theta}|a)}(\boldsymbol{\theta}-\hat{\boldsymbol{\theta}})^2$ to be negative as $a \to \infty$. In what follows, we will show some examples with specific distributions that are widely used in extreme value analysis where this happens. We conjecture that this is true for a broad set of cumulative distribution functions, since it worked in all the examples we tried. This would be an interesting open problem to solve in the future.

\begin{example} {\bf Exponential distribution}\\
The density of the exponential distribution is given by
$$f(x|\lambda) = \frac{1}{\lambda}e^{-x/\lambda}$$
where $x \ge 0$ and $\lambda > 0$. So the tail probability is
$$1-F(a|\lambda)=\int^{\infty}_a f(x|\lambda)dx=e^{-\frac{a}{\lambda}}$$
Taking derivative with respect to $\lambda$ at both sides we have
$$
-\frac{d }{d \lambda} F(a|\lambda)=\frac{a}{\lambda^2}e^{-\frac{a}{\lambda}};\ \
-\frac{d^2 }{d \lambda^2} F(a|\lambda) =\left(-\frac{2a}{\lambda^3}+\frac{a^2}{\lambda^4}\right)e^{-\frac{a}{\lambda}}; \ \
-\frac{d^3 }{d \lambda^3} F(a|\lambda) =\left(\frac{6a}{\lambda^4}-\frac{6a^2}{\lambda^5}+\frac{a^3}{\lambda^6}\right)e^{-\frac{a}{\lambda}}
$$
Suppose we have i.i.d sample ${\bf x}=(x_1, ..., x_n)$ from $f(x|\lambda)$, then the marginal distribution can be calculated as 
\begin{equation*}
m({\bf x})=\int^{\infty}_0f({\bf x}|\lambda)\pi_J(\lambda)d\lambda = \Gamma(n)\Big(\sum_{i=1}^{n} x_i \Big) ^{-n},
\end{equation*}
where we use Jeffereys prior as $\pi_J(\lambda) \propto 1/ \lambda$.
By which the posterior distribution is obtained as
$$\pi(\lambda|{\bf x})=\frac{1}{\lambda^{(n+1)}\Gamma(n)}\exp\left(-\frac{\sum^n_{i=1}x_i}{\lambda}\right)\left(\sum^n_{i=1}x_i\right)^n$$
After some arithmetic manipulation and using $\hat{\lambda}=\bar{x}$ we obtain
$$\begin{aligned}
E^{\pi(\lambda|{\bf x})}(\lambda-\hat\lambda) & = \int_{0}^{\infty}(\lambda-\hat\lambda) \pi(\lambda|{\bf x}) d\lambda=\frac{ \bar x}{(n-1)},\\
E^{\pi(\lambda|{\bf x})}(\lambda-\hat\lambda)^2 & = 
\int_{0}^{\infty}(\lambda-\hat\lambda)^2 \pi(\lambda|{\bf x}) d\lambda
=\bar x^2  \frac{n+2}{(n-1)(n-2)},\\
E^{\pi(\lambda|{\bf x})}(\lambda-\hat\lambda)^3 & = 
\int_{0}^{\infty}(\lambda-\hat\lambda)^3 \pi(\lambda|{\bf x}) d\lambda
 =\bar x^3  \frac{7n+6}{(n-1)(n-2)(n-3)}.
\end{aligned}$$
Plug these terms into $D(a)$ we have
$$\small \begin{aligned}
D(a) & = -\left.\frac{d }{d \lambda} F(a|\lambda)\right|_{\hat{\lambda}}E^{\pi(\lambda|{\bf x})}(\lambda-\hat{\lambda}) - \frac{1}{2}\left.\frac{d^2 }{d \lambda^2} F(a|\lambda)\right|_{\hat{\lambda}} E^{\pi(\lambda|{\bf x})}(\lambda-\hat{\lambda})^2 -\frac{1}{6}\left.\frac{d^3 }{d \lambda^3} F(a|\lambda)\right|_{{\lambda}_L}E^{\pi(\lambda|{\bf x})} (\lambda-\hat{\lambda})^3\\
& =e^{-\frac{a}{\bar{x}}}\left[\frac{a}{\bar{x}}\frac{-4}{(n-1)(n-2)}+\frac{a^2}{2\bar{x}^2}\frac{n+2}{(n-1)(n-2)}\right]+e^{-\frac{a}{\lambda_L}}\left(\frac{a^3-6\lambda_La^2+6\lambda_L^2a}{6\lambda_L^6}\right)\frac{\bar{x}^3(7n+6)}{(n-1)(n-2)(n-3)}
\end{aligned}$$
Here we have that $a > >0$ , and $\bar{x} \geq 0$; $\lambda_L$ is some number between $x$ and $\bar{x}$ so $\lambda_L \geq 0$. All of which implies $D(a) \geq 0$.
\\
Then, we want to show that $\lim_{a \to \infty}R^*(\lambda)=0$, it is sufficient to show that $$\lim_{a \to \infty}\left.\frac{d^3 F(a|\lambda)}{d \lambda^3}\right|_{\lambda_L}=0.$$
And we could obtain this simply by using L'Hospital's rule. In conclusion, the second derivatives for exponential distribution $\exp(\lambda)$ w.r.t. $\lambda$ is
$$\begin{aligned}
\frac{d^2}{d\lambda^2}F(a|\lambda) & =\left(2-\frac{a}{\lambda}\right)\frac{a}{\lambda^3}e^{-\frac{a}{\lambda}}
\end{aligned}$$
Since $a$ is assumed to be some extreme number, which implies $d^2 F(a|\lambda) / d\lambda^2 \leq 0$, i.e. the tail probability $\varphi(\lambda)=1-F(a|\lambda)$ is convex.

\end{example}

\begin{example} {\bf Pareto Distribution}\\
Given the scale parameter $\beta=1$ and the shape parameter $\alpha$ as unknown, the pdf of the Pareto distribution is given by
$$f(x|\alpha) = \alpha x^{-\alpha - 1}$$
where $x \ge 1$ and $0 < \alpha < 1$, and the cumulative distribution function is
$$F(x|\alpha) =\int^{x}_1 f(t|\alpha)dt = 1 - x^{-\alpha}, x \geq 1 $$
By setting the derivative of the log-likelihood equal to zero we get the MLE of $\alpha$ as
$\hat{\alpha}=n/ \sum^n_{i=1}\log x_i$. We are interested in calculating the tail probability when $b$ is extremely large. Note that
$$\varphi(\alpha)=1 - F(b|\alpha) = b^{-\alpha} $$
Taking derivatives of $\varphi(\alpha)$ with respect to $\alpha$ we obtain
$$
-\frac{d }{d \alpha}F(b|\alpha) =-b^{-\alpha} \ln b; \ \
-\frac{d^2 }{d \alpha^2}F(b|\alpha) =b^{-\alpha} (\ln b)^2; \ \
-\frac{d^3 }{d \alpha^3}F(b|\alpha) = -b^{-\alpha} (\ln b)^3.
$$
Using Jeffreys's prior $\pi_J(\alpha) \propto 1/ \alpha$, we have
$$
m({\bf x}) =\int^1_0 f({\bf x}|\alpha) \pi_J(\alpha) d \alpha
= \frac{\Gamma(n,0)-\Gamma(n,\sum^n_{i=1}\ln x_i)}{\prod^n_{i=1}x_i \left[\sum^n_{i=1}\ln x_i\right]^n},
$$
where the upper incomplete gamma function is defined as $\Gamma (s,x)=\int _{x}^{\infty }t^{s-1}e^{-t}dt$. Then the posterior distribution is given by
$$\begin{aligned}
\pi(\alpha|{\bf x}) & =\frac{L(\alpha|{\bf x})\pi_J(\alpha)}{m({\bf x})}=\frac{\alpha^{n-1}\left(\prod^n_{i=1}x_i\right)^{-\alpha}\left[\sum^n_{i=1}\ln x_i \right]^n}{\Gamma(n,0)-\Gamma(n,\sum^n_{i=1}\ln x_i)}
\end{aligned}$$
Using the properties of the incomplete gamma function, and integration by parts we find the recurrence relation $\Gamma(s+1,x)=s\Gamma(s,x)+x^se^{-x}$. We obtain
$$\small \begin{aligned}
E^{\pi(\alpha|{\bf x})}(\alpha-\hat{\alpha}) & =\int^1_0 (\alpha-\hat{\alpha}) \pi(\alpha|{\bf x}) d \alpha = -\frac{\left[\sum^n_{i=1}\ln x_i \right]^{n-1}}{\prod^n_{i=1}x_i[\Gamma(n,0)-\Gamma(n,\sum^n_{i=1}\ln x_i)]},\\
E^{\pi(\alpha|{\bf x})}(\alpha-\hat{\alpha})^2 & =\frac{n}{\left[\sum^n_{i=1}\ln x_i \right]^2}+\frac{(n-1)(\sum^n_{i=1}\ln x_i)^{n-2}-(\sum^n_{i=1}\ln x_i)^{n-1}}{\prod^n_{i=1}x_i[\Gamma(n,0)-\Gamma(n,\sum^n_{i=1}\ln x_i)]},\\
E^{\pi(\alpha|{\bf x})}(\alpha-\hat{\alpha})^3 &=\frac{2n}{(\sum^n_{i=1}\ln x_i)^3}-\frac{(\sum^n_{i=1}\ln x_i)^{n-3}[n^2+2+(2-2n)(\sum^n_{i=1}\ln x_i)+(\sum^n_{i=1}\ln x_i)^2]}{(\prod^n_{i=1}x_i)[\Gamma(n,0)-\Gamma(n,\sum^n_{i=1}\ln x_i)]}.
\end{aligned}$$
To show $D(b)\geq 0$, is equivalent to show that the first term in the expression of $D(b)$ after plugging in the Taylor expansion of $1-F(b|\alpha)$ goes to zero as $b \to \infty$, which could be obtain by using the L'Hospital's rule. And we also need to show the second term $\left.d^2 /d \alpha^2 F(b|\alpha)\right|_{\hat{\alpha}}E^{\pi(\alpha|{\bf x})}(\alpha-\hat{\alpha})^2$ is asymptotically negative. We can see this from the fact that $\left.d^2 /d \alpha^2 F(b|\alpha)\right|_{\hat{\alpha}}=-b^{-\alpha} (\ln b)^2 \leq 0$ and $E^{\pi(\alpha|{\bf x})}(\alpha-\hat{\alpha})^2 \geq 0$.
\\ \\
Then, We want to show that $\lim_{b \to \infty}R^*(\alpha)=0$, it is sufficient to show that $$\lim_{b \to \infty}\left.\frac{d^3}{d \alpha^3} F(b|\alpha)\right|_{\alpha_L}=0.$$
And we could obtain this simply by using L'Hospital's rule. In conclusion, the second derivatives for Pareto distribution is
$$\begin{aligned}
\frac{d^2}{d \alpha^2}F(b|\alpha) & =-b^{-\alpha}(\ln b)^2
\end{aligned}$$
Since $b$ is assumed to be some extreme number, which implies $d^2 /d \alpha^2 F(b|\alpha) \leq 0$, i.e. the tail probability $\varphi(\alpha)=1-F(b|\alpha)$ is convex.
\end{example}

\begin{example}{\bf Normal Distribution}\\
Normal distribution with unknown standard deviation $\sigma$ and expectation $\mu$ is a case where the parameter is a two dimensional vector, i.e. $\boldsymbol{\theta}=(\mu, \sigma)$. Since $x|\mu, \sigma \sim N(\mu,\sigma^2)$, then the CDF when $x=a$ is
$$F(a|\mu, \sigma)=\frac{1}{2}\left[1+\erf\left(\frac{a-\mu}{\sigma\sqrt{2}}\right)\right]=\frac{1}{2}+\frac{1}{\sqrt{\pi}}\int^{\frac{a-\mu}{\sigma\sqrt{2}}}_0e^{-t^2}dt$$
where $\erf(x)$ is the related error function defined as $\erf(x)=2 / \sqrt{\pi} \int^x_0e^{-t^2}dt$. 
Looking at the Hessian matrix, we have
$$\begin{aligned}H &=\displaystyle{\begin{bmatrix} \frac{\partial^2}{\partial \mu^2}F(a|\mu, \sigma) & \frac{\partial^2}{\partial \sigma \partial \mu}F(a|\mu, \sigma) \\ \frac{\partial^2}{\partial \mu \partial \sigma}F(a|\mu, \sigma) & \frac{\partial^2}{\partial \sigma^2}F(a|\mu, \sigma) \end{bmatrix}}\\
& = \begin{bmatrix} -\frac{a-\mu}{\sqrt{2\pi}\sigma^3}e^{-\left(\frac{a-\mu}{\sigma\sqrt{2}}\right)^2} & -\frac{1}{\sqrt{2\pi}\sigma^2}e^{-\left(\frac{a-\mu}{\sigma\sqrt{2}}\right)^2}\left[\frac{(a-\mu)^2}{\sigma^2}-1\right]\\
-\frac{1}{\sqrt{2\pi}\sigma^2}e^{-\left(\frac{a-\mu}{\sigma\sqrt{2}}\right)^2}\left[\frac{(a-\mu)^2}{\sigma^2}-1\right] &
-\frac{a-\mu}{\sqrt{2\pi}\sigma^3}e^{-\left(\frac{a-\mu}{\sigma\sqrt{2}}\right)^2}\left[\frac{(a-\mu)^2}{\sigma^2}-2\right]
\end{bmatrix}
\end{aligned}$$
To show $H$ is negative definite for large $a$, we need to show for $\forall {\bf v}^T=(v_1,v_2)\neq {\bf 0}$, we have 
$${\bf v}^TH{\bf v} <0.$$
By tedious calculation we have 
$${\bf v}^TH{\bf v}=-\frac{1}{\sqrt{2\pi}\sigma^2}e^{-\left(\frac{a-\mu}{\sigma\sqrt{2}}\right)^2}\left[ (v_1^2-2v_2^2)\left(\frac{a-\mu}{\sigma}\right)+2v_1v_2\left(\frac{a-\mu}{\sigma}\right)^2+v_2^2\left(\frac{a-\mu}{\sigma}\right)^3-2v_1v_2 \right]$$
Since $a$ is assumed to be some extreme large number, so $a-\mu>0$, then the leading term in the bracket is $$v_2^2\left(\frac{a-\mu}{\sigma}\right)^3$$
which is positive. Hence, ${\bf v}^TH{\bf v} <0$, i.e. $H$ is negative definite for large $a$ as we expected. In other words, $\varphi(\mu, \sigma)=1-F(a|\mu, \sigma)$ is a convex function.
\end{example}

\begin{figure}
  \begin{subfigure}{5cm}
    \centering\includegraphics[width=5cm]{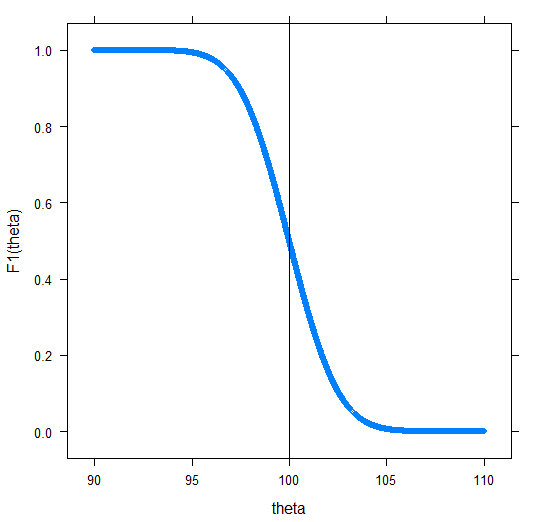}
    \caption{Normal Distribution}
  \end{subfigure}
  \begin{subfigure}{5cm}
    \centering\includegraphics[width=5cm]{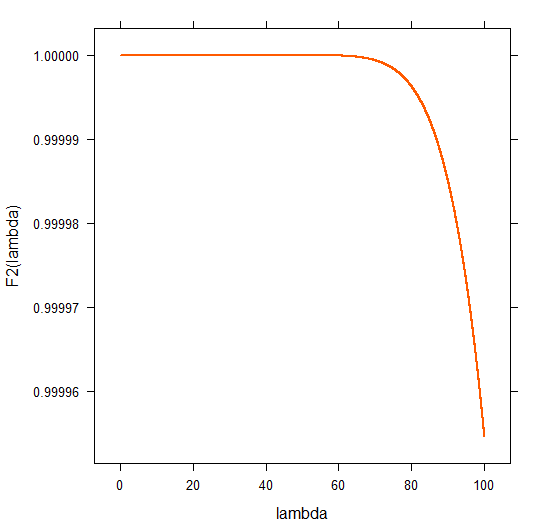}
    \caption{Exponential Distribution}
  \end{subfigure}
  \begin{subfigure}{5cm}
    \centering\includegraphics[width=5cm]{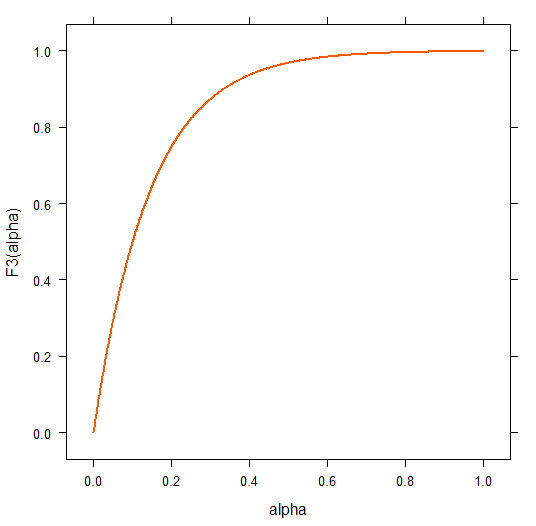}
    \caption{Pareto Distribution}
  \end{subfigure}
\caption{Plots of distribution functions when parameters are variables. Note in (a), the parameter $\theta$ is the mean of the normal distribution, i.e. in this case $\sigma$ is given and we can see that $F(a|\theta)$ is concave down when $a>\theta$.  }
\end{figure}

\section{Conclusions}
Bayesian and frequentist estimations of the tail probability sometimes give conclusions of huge differences which can cause serious consequences. Thus in practice, we will have to take both of the results into consideration. To be specific, Bayesian method always estimate the tail probability higher than the frequentist model. The Bayesian estimation of probability of tails is well founded on Probability Theory: It is a marginal computation that integrates out the parameters of the tail. On the other hand, the frequentist estimation is an approximation.

By looking at the Taylor expansion of the tail and investigate the convexity of the distribution function, we claim that the Bayesian estimator for tail probability being higher than the frequentist estimator depends on how $\varphi(\boldsymbol{\theta})=1-F(a|\boldsymbol{\theta})$ is shaped. The condition that $\varphi(\boldsymbol{\theta})$ is a strictly convex function is equivalent to  $H_{\boldsymbol{\theta}}F(a|\boldsymbol{\theta})<0$. Other examples (only continuous cases with infinite support) like the Cauchy Distribution, Logistic Distribution, Log-normal Distribution, Double Exponential Distribution, Weibull Distribution, etc., also satisfy our convexity conditions here.

However, in general convexity is a much stronger argument than Jensen's inequality, i.e. $\varphi(\boldsymbol{\theta})=1-F(a|\boldsymbol{\theta})$ is a convex function, or equivalently $H_{\boldsymbol{\theta}}F(a|\boldsymbol{\theta})<0$ is only a sufficient condition for Jensen's inequality to hold but not a necessary condition. There are distributions with $H_{\boldsymbol{\theta}}F(a|\boldsymbol{\theta}) \geq 0$ but we still have the Bayesian estimator for the tail probability is higher than the frequentist approximation. In \cite{iffjensen}, they found conditions on the random variable to make the other direction work which will be discussed in our future work.

\section{Acknowledgment}
We grateful thank the Royal Statistical Society for giving us the permission to reuse the Figure \ref{cp1} and \ref{cp2} of Coles and Pericchi \cite{cp} Applied Statistics(2003), Volume52, Issue4, Pages 405-416, published by Wiley.

%
%
%
%
%


\end{document}